\documentclass[11pt]{article}

\usepackage[
	persons={nima,michal,june,elizabeth},
	authors=blocks,
	tags=noreview,
	bibliosources={refs.bib}
]{pomegranate}

\AddTag<noreview>
\Tag<review>{\RemoveTag<noreview>}

\DeclareOperator{\PP}
\DeclareOperator{\TV}
\DeclareOperator{\mix}
\DeclareOperator{\Cov}
\DeclareOperator{\polylog}
\DeclareOperator{\KL}
\DeclareOperator{\diag}
\DeclareOperator{\sign}
\DeclareOperator{\Var}
\DeclareOperator{\Eng}
\DeclareDocumentMathCommand{\DKL}{}{\D_\KL}
\DeclareDocumentMathCommand{\corMat}{}{\Psi^{\text{cor}}}

\DeclareMathOperator{\calD}{\mathcal{D}}

\title{Domain Sparsification of Discrete Distributions using Entropic Independence}
\Tag<noreview>{
	\author[1]{Nima Anari}
	\author[2]{Micha\l\ Derezi\'nski}
	\author[1]{Thuy-Duong Vuong}
	\author[2]{Elizabeth Yang}
	\affil[1]{Stanford University, \url{{anari,tdvuong}@stanford.edu}}
	\affil[2]{UC Berkeley, \url{{mderezin,elizabeth_yang}@berkeley.edu}}
}
\Tag<review>{
	\author{}
}
\date{}

\begin{document}
	\maketitle
	
	\begin{abstract}
		We present a framework for speeding up the time it takes to sample from discrete distributions $\mu$ defined over subsets of size $k$ of a ground set of $n$ elements, in the regime where $k$ is much smaller than $n$. We show that if one has access to estimates of marginals $\P_{S\sim \mu}{i\in S}$, then the task of sampling from $\mu$ can be reduced to sampling from related distributions $\nu$ supported on size $k$ subsets of a ground set of only $n^{1-\alpha}\cdot \poly(k)$ elements. Here, $1/\alpha\in [1, k]$ is the parameter of entropic independence for $\mu$. Further, our algorithm only requires sparsified distributions $\nu$ that are obtained by applying a sparse (mostly $0$) external field to $\mu$, an operation that for many distributions $\mu$ of interest, retains algorithmic tractability of sampling from $\nu$. This phenomenon, which we dub domain sparsification, allows us to pay a one-time cost of estimating the marginals of $\mu$,  and in return reduce the amortized cost needed to produce many samples from the distribution $\mu$, as is often needed in upstream tasks such as counting and inference.
		
		For a wide range of distributions where $\alpha=\Omega(1)$, our result reduces the domain size, and as a corollary, the cost-per-sample, by a $\poly(n)$ factor. Examples include monomers in a monomer-dimer system, non-symmetric determinantal point processes, and partition-constrained Strongly Rayleigh measures. Our work significantly extends the reach of prior work of Anari and Derezi\'nski who obtained domain sparsification for distributions with a log-concave generating polynomial (corresponding to $\alpha=1$). As a corollary of our new analysis techniques, we also obtain a less stringent requirement on the accuracy of marginal estimates even for the case of log-concave polynomials; roughly speaking, we show that constant-factor approximation is enough for domain sparsification, improving over $O(1/k)$ relative error established in prior work.
	\end{abstract}

	\section{Introduction}

Sparsification has been a crucial idea in designing many fast algorithms; famous examples include cut or spectral graph sparsifiers \cite{ST11} and dimension reduction using sparsified Johnson-Lindenstrauss transforms \cite{DKS10}. In this work, we address the question of sparsifying discrete distributions, with the goal of speeding up the fundamental task associated with distributions: sampling from them.

As an illustrative example building towards our notion of sparsification for distributions, consider the task of sampling a (uniformly) random edge in a graph. Suppose that we have a graph on $n$ non-isolated vertices, and we are allowed to make adjacency queries. How many of the $n$ vertices do we have to ``look at'' before we observe both endpoints of some edge? The answer to this question depends on the structure of the graph; for a star graph, where a single vertex is connected to the other $n-1$, we would have to find this central vertex to have any chance of observing an edge; so no amount of smart guessing can result in looking at $\ll n$ vertices. However, for regular graphs, where every vertex has the same degree, because of the Birthday Paradox phenomenon, it is enough to look at a sample of $O(\sqrt{n})$ vertices picked uniformly at random to observe an edge between two of them with overwhelming probability. The bound of $O(\sqrt{n})$ is indeed the best possible, since in a random perfect matching (degree $1$ regular graph), the best and only sensible strategy is to pick vertices at random.
 
This curious phenomenon generalizes to hypergraphs as well. On a $k$-uniform hypergraph, with hyperedges representing sets of $k$ vertices, to observe a hyperedge one has to generally look at $\simeq n$ vertices in the worst case. But on regular hypergraphs, a substantially smaller sample, namely $\simeq n^{1-1/k}$ many vertices, will contain a hyperedge with high probability \cite[see, e.g.,][for forms of Birthday Paradox related to $k$-sets and $k$-collisions]{KDKKSBB06}. Notice that this improvement quickly deteriorates as $k$ gets large, and becomes meaningless as soon as $k\simeq \log n$.

When can the bound of $n^{1-1/k}$ for regular hypergraphs be improved, ideally to a polynomially small fraction of the $n$ vertices, even for $k\gg 1$? Moreover, suppose that instead of desiring just one of the hyperedges, we want to extract an (approximately) \emph{uniformly random hyperedge}. Can we still produce a hyperedge following this distribution, by only looking at a small subset of vertices? In a nutshell, how small of a (random) vertex set can we look at in order to have a distributionally representative ``sparsification'' of the entire hypergraph?

In this work, we tie the answer to these questions to notions of high-dimensional expansion, specifically the notion of \emph{entropic independence} introduced by \textcite{AJKPV21}. To every measure, a.k.a.\ weighted hypergraph, on size $k$ subsets of $\set{1,\dots,n}$ denoted by $\mu:\binom{[n]}{k}\to\R_{\geq 0}$, one can associate a parameter of entropic independence $1/\alpha\in[1,k]$, defined formally in \cref{sec:prelims}. A larger $\alpha$ corresponds to better high-dimensional expansion. We show that the hypergraph defined by $\mu$, while being truly $k$-uniform, behaves almost as if it was $(1/\alpha)$-uniform: informally, we can ``sparsify'' this hypergraph by looking at only $n^{1-\alpha}\cdot\poly(k)$ vertices, under some ``regularity assumptions.'' To avoid confusion with other classical concepts of graph and hypergraph sparsification, which primarily keep the vertex set while deleting a subset of the edges, we call this type of sparsification \emph{domain sparsification}.

A long line of recent works have obtained breakthroughs in sampling and counting by viewing combinatorial distributions as (weighted) hypergraphs and studying notions of high-dimensional expansion for them \cite{ALOV19,CGM19,AL20,ALO20,CLV20,GM20,ALOVV20,CGSV21,FGYZ21,AASV21,Liu21,BCCPSV21,JPV21,AJKPV21,ALO21,CFYZ21}. A central theme in all of the aforementioned works is the establishment of some form of high-dimensional expansion for a hypergraph encoding the probability distribution of interest. At a high-level, these notions can be viewed as measures of proximity to independent/product distributions. Sampling from distributions extremely close to product distributions is roughly as easy as sampling i.i.d.\ from the marginal distribution over single elements; it is no surprise then, that for distributions with limited correlations, knowledge of marginals can boost sampling time. This is what we formally establish in this work.

Our main result applies to distributions that have \emph{entropic independence} \cite{AJKPV21}, a notion stronger than \emph{spectral independence} \cite{ALO20}, but weaker than \emph{fractional log-concavity} and \emph{sector-stability} \cite{AASV21}. Roughly speaking, a background measures $\mu$ over $k$-sized sets is entropically independent, if for any (randomly chosen) set $S$, the relative entropy of a uniformly random element of $S$ is at most $1/\alpha k$ fraction of the relative entropy of $S$, where we usually take $\alpha=\Omega(1)$. The main intuition leading to our results is that high correlations in such distributions must be limited to small groups of elements. By sampling enough many elements from the domain, we cover these correlated groups.

Similar to graph sparsification, in domain sparsification we need to reweigh the sparsified object. This is achieved by the standard operation of applying an external field. For a weight vector $\lambda\in \R_{\geq 0}^n$, the $\lambda$-external field applied to $\mu$ is the distribution $\lambda\star \mu$ defined by
\[ \lambda\star\mu(S)\propto \mu(S)\prod_{i\in S}\lambda_i. \]

\begin{theorem}[Informal] \label{thm:main}
Let $\mu: \binom{[n]}{k} \to \mathbb{R}_{\geq 0}$ be $(1/\alpha)$-entropically independent. Suppose that we have access to estimates $p_1,\dots,p_n$ of the marginals and an oracle that can produce i.i.d.\ samples $i\in [n]$ with $\P{i}\propto p_i$; suppose that our estimates satisfy $p_1+\dots+p_n=k$ and
\[ p_i \geq \Omega(\P{i\in S})  \]
for all $i$. Then we can produce a random sparse external field $\lambda\in \R_{\geq 0}^n$ with at most $n^{1-\alpha}\cdot \poly(k)$ nonzero entries, in time $n^{1-\alpha}\cdot \poly(k)$, such that a random sample $S$ of $\lambda\star \mu$ approximately follows the distribution defined by $\mu$.
\end{theorem}

\cref{thm:main} follows directly from \cref{prop:near-isotropic,prop:fast sampling,lem:subsample lower bound}. We outline our techniques in \cref{subsec:outline}.

Notice that we do not need degree-regularity of $\mu$, which would be equivalent to $\P_{S\sim \mu}{i\in S}$ being \emph{exactly} the same for all $i$. Instead, it is enough to just have an \emph{estimate} of these marginals, a weaker condition than regularity. This is because instead of sampling a subset of vertices uniformly at random, we can sample a biased subset of vertices, with probability biases defined by the marginals, and domain sparsification will still work.

Prior to our work, domain sparsification was known for distributions with log-concave generating polynomials (the case of $\alpha=1$) \cite{AD20}, based on techniques inspired by earlier algorithms for sampling from determinantal point processes (an even narrower class) \cite{Der19,DCV19}. All of these distributions satisfy forms of negative dependence \cite{BBL09,AD20} that were crucial in obtaining domain sparsification for them. Our work significantly extends the reach of domain sparsification beyond these classes; as we will see, for \emph{any distribution} $\mu$, we have $\alpha\geq 1/k$, and as a simple corollary we get nontrivial domain sparsification for \emph{any distribution} $\mu$ as long as $k=O(1)$, a result which appears to be nontrivial on its own.

The main application of domain sparsification is in accelerating the time it takes to produce multiple samples from a distribution $\mu$. Suppose that an algorithm $\mathcal{A}$ can produce (approximate) samples from a distribution $\mu$ and any distribution obtained from it by an external field, in time $T(n, k)$, which usually depends polynomially on $n$.\footnote{Typically the running time has logarithmic dependencies on the approximation error and potentially magnitude of external fields, but for simplicity of exposition we hide them here.} Then after a preprocessing step, where we use $\mathcal{A}$ to estimate the marginals of $\mu$, we can produce new samples in time $T(n^{1-\alpha}\cdot \poly(k), k)$ per sample, which is polynomially smaller than $T(n, k)$, as long as $k$ is smaller than some $\poly(n)$ threshold. Notice that the preprocessing step has to be done only once, and its cost gets amortized when we are interested in obtaining multiple samples from $\mu$. A careful implementation, directly adapted from what was done for log-concave polynomials by \textcite{AD20}, can bootstrap domain sparsification with estimation of marginals to complete the preprocessing step in roughly $\simeq T(n, k)+n\cdot \poly(k, \log n)\cdot T(n^{1-\alpha}\cdot \poly(k), k)$ time.

\begin{corollary}[Informal, adapted from \cite{AD20}]\label{cor:marginals}
	Suppose that we have an algorithm $\mathcal{A}$ that can produce approximate samples from any external field $\lambda$ applied to $\mu$ in time $T(m, k)$, where $m$ is the sparsity of $\lambda$. Then we can produce the marginal estimates $p_i$ and the i.i.d.\ sampling oracle required in \cref{thm:main} in time
	\[ O\parens*{T(n, k)+n\cdot \poly(k, \log n)\cdot T(n^{1-\alpha}\cdot \poly(k), k)}. \]
	
	Further, for any desired $t$, we can produce $t$ i.i.d.\ approximate samples from $\mu$ in time
	\[ O\parens*{T(n, k)+\max\set{t, n\cdot \poly(k, \log n)}\cdot T(n^{1-\alpha}\cdot \poly(k), k)}. \] 
\end{corollary}

Sampling is often used to solve the problem of approximate counting, that is computing the partition function
\[ \sum_{S}\mu(S). \]
To obtain an $\epsilon$-relative error approximation, known reductions between counting and sampling \cite{JVV86} introduce at least a multiplicative factor of $1/\epsilon^2$ to the sampling time. Directly adapting the same technique for log-concave polynomials \cite{AD20} and combining with our new domain sparsification result, we obtain an $\epsilon$-relative error of the counts in time $\simeq T(n, k)+\max\set{n, 1/\epsilon^2}\cdot \poly(k, \log n)\cdot T(n^{1-\alpha}\cdot \poly(k), k)$. Notice that here $1/\epsilon^2$ is multiplied by the term $T(n^{1-\alpha}\cdot \poly(k), k)$ that can be substantially smaller than $T(n, k)$; as a result, we can get a substantially improved running time for the high-precision regime where $\epsilon$ is inverse-polynomially small.

\begin{corollary}[Informal, adapted from \cite{AD20}]\label{cor:counting}
	Suppose that we have an algorithm $\mathcal{A}$ that can produce approximate samples from any external field $\lambda$ applied to $\mu$ in time $T(m, k)$, where $m$ is the sparsity of $\lambda$. Then we can compute an $\epsilon$ relative error approximation of $\sum_{S}\mu(S)$ in time
	\[ O\parens*{T(n, k)+\max\set{n,1/\epsilon^2}\cdot \poly(k, \log n)T(n^{1-\alpha}\cdot \poly(k), k)}. \] 
\end{corollary}

\begin{remark}
	For many applications, we can derive entropic independence of $\mu$ from a stronger property called fractional log-concavity \cite{AASV21,AJKPV21}. For an $\alpha$-fractionally log-concave distribution, recent work of \textcite{AJKPV21} established Modified Log-Sobolev Inequalities for natural (multi-step) down-up random walks. For simplicity of exposition, assume that $1/\alpha\in \Z$ and $\alpha=\Omega(1)$. Then, these random walks produce approximate samples from $\mu$ in the following number of steps:
	\[ O\parens*{k^{1/\alpha}\cdot \log \log \frac{1}{\P_{\mu}{S_0}}}, \]
	where $S_0$ is the starting point of the random walk. Further, each step of the random walk requires querying $\mu$ at $n^{1/\alpha}$ points, leading to a total runtime of
	\[ O\parens*{(kn)^{1/\alpha}\cdot \log \log \frac{1}{\P_{\mu}{S_0}}}. \]
	In most settings, such as when the bit-complexity of $\mu$ is bounded by $\poly(n)$, the extra $\log \log(1/\P_{\mu}{S_0})$ can be safely ignored, as long as we make sure $S_0$ is in the support. So one can think of this Markov chain as an algorithm $\mathcal{A}$ that, up to this initial step of finding a suitable starting point, satisfies
	\[T(n, k)=\tilde O\parens*{(nk)^{1/\alpha}}.\]
	Note that for this choice of the algorithm $\mathcal{A}$ and running time $T$, the bounds in \cref{cor:marginals,cor:counting} simplify as 
	\[ n \cdot \poly(k, \log n)\cdot T(n^{1-\alpha}\cdot \poly(k), k)\simeq \poly(k, \log n)\cdot T(n, k). \]
	However, our results apply to any choice of a base sampling algorithm $\mathcal{A}$.
\end{remark}

A challenging part of obtaining our results is the lack of negative dependence inequalities, which were used by the prior work of \textcite{AD20}. These negative dependence inequalities result in domain sparsification with sparsified domain size solely depending on $k$, with no dependence on $n$. We show in \cref{sec:lowerbound} that our analysis of our domain sparsification scheme is tight. An intriguing question is if we can find other domain sparsification schemes, perhaps using higher-order marginals, that sparsify domains to size $\poly(k, \log n)$? We make the following conjecture.

\begin{conjecture}[Informal]\label{conj:high-order}
Let $\mu$ be an $\alpha$-fractionally-log-concave distribution for some $\alpha=\Omega(1)$. Given  access to estimates for high-order marginals of  the form $\P_{S \sim \mu}{T \subseteq  S}$ for all $T$ of size $\l\simeq 1/\alpha$, and an oracle that produces i.i.d.\ samples from these marginals, there is a domain sparsification scheme for $\mu$ which reduces the domain size to only $\poly(k)$.
\end{conjecture}

Despite the attractiveness of a bound independent of $n$, we give evidence that obtaining these domain sparsification schemes requires entirely new ideas; we show in \cref{sec:lowerbound} that if we replace fractional log-concavity by entropic independence (which is sufficient for our main result, \cref{thm:main}) in the above conjecture, the conjecture becomes false.

\subsection{Applications}\label{subsec:applications}

Here we mention examples of distributions to which our results can be applied beyond those covered by prior work of \textcite{AD20}. Our examples satisfy fractional log-concavity \cite{AASV21} which entails both entropic independence \cite{AJKPV21}, and the existence of the base sampling algorithm $\mathcal{A}$.

For a distribution $\mu:\binom{[n]}{k}\to\R_{\geq 0}$ we define the generating polynomial $g_\mu$ to be
\[ g_\mu(z_1,\dots,z_n):=\sum_{S}\mu(S)\prod_{i\in S}z_i. \]
We say that the distribution $\mu$ or the polynomial $g_\mu$ is $\alpha$-fractionally-log-concave for some parameter $\alpha\leq 1$ if $\log g_\mu(z_1^\alpha,\dots,z_n^\alpha)$ is concave  as a function over the positive orthant $(z_1,\dots,z_n)\in \R_{\geq 0}^n$.

Notice that any multi-affine homogeneous polynomial with nonnegative coefficients is the generating polynomial of a distribution $\mu$. Throughout the paper, we often equate these polynomials with the distributions they represent, and freely talk about fractional log-concavity of either the generating polynomial or the distribution. For more details, see \cite{AASV21}.

\begin{example}
\par If $g$ is a degree-$k$ homogeneous multi-affine polynomial, then it is $\frac{1}{k}$-log-concave. Every monomial $\prod_{i\in S} z_i^{1/k}$ is concave, since by H\"older's inequality
\[\prod_{i\in S} (\lambda z_i + (1-\lambda) y_i) \geq \parens*{\lambda  \prod_{i\in S} z_i^{1/k} + (1-\lambda)\prod_{i\in S} y_i^{1/k}}^k.\] 
Now, $g(z_1^{1/k}, \dots, z_n^{1/k})$ is concave as (weighted)-sum of concave functions $\prod_{i\in S} z_i^{1/k}$. Thus is also log-concave, as $\log$ is a monotone and concave function.
\end{example}

\begin{example}\label{ex:blowup}
We present another toy class of $\alpha$-fractionally-log-concave polynomials that provides some intuition despite not having many applications. Let $\mu$ be an $\alpha$-fractionally-log-concave polynomial over the variables $z_1, \ldots, z_n$. If we replace each $z_i$ with the monomial $\prod_{j = 1}^m z_i^{(j)}$, we obtain a degree $mk$, $\frac{\alpha}{m}$-fractionally-log-concave polynomial over the variables $\set{z_i^{(j)}\given i\in [n], j\in [m]}$ \cite{AASV21}. For example, if the starting distribution $\mu$ is the uniform distribution over bases of a matroid, then $\alpha=1$, and the resulting distribution will be $1/m$-fractionally log-concave.

Notice that if we normalize this blown-up polynomial to convert it into to a distribution, for any $i \in [n]$, the elements $i^{(1)}, \ldots, i^{(m)}$ are all perfectly correlated. On the other hand, if $i \neq j$, any two elements $i^{(m_i)}, j^{(m_j)}$ inherit the correlations from the log-concave distribution. 
\end{example}

\begin{example} \label{ex:matching}
Let $G$ be a graph and $k\in \N$. For each set $S \subseteq \binom{V}{2k}$, set $\mu(S)$ to be proportional to the number of perfect matchings on $S$. Sampling from $\mu$ allow us to approximately count the number of $k$-matching, i.e., matchings using $k$ edges. \cite{AASV21} proved that for any value of $k$, this distribution is fractionally log-concave with $\alpha\geq 1/4$.

Not all choices of $G$ result in efficient sampling algorithms. The implementation of each iteration of the Markov chain involves counting perfect matchings over $S \subseteq V$, and we do not have a $\poly(k)$ time algorithm for counting matchings in general graphs. We thus only consider downward closed graph families with an FPRAS for counting perfect matchings, e.g., bipartite graphs \cite{JSV04}, planar graphs \cite{Kas67}, certain minor-free graphs \cite{EV19}, and small genus graphs \cite{GL99}. Our main results imply that as long as we estimate the probability of every vertex being part of a random $k$-matching, we can reduce the task of sampling $k$-matchings on an $n$ vertex graph to graphs with only $n^{3/4}\cdot \poly(k)$ many vertices.
\end{example}

\begin{example} \label{ex:DPP}
Let $L$ be a nonsymmetric positive semidefinite matrix, i.e., an $n\times n$ matrix $L$ that satisfies $L+L^\intercal\succeq 0$. Then, the nonsymmetric $k$-determinantal point process \cite[see, e.g.,][]{GBDK19,GHDGB20,AV21} with kernel $L$, defined by 
\[\mu(S) = \det(L_{S,S})\]
for all $S\in \binom{[n]}{k}$ is fractionally log-concave with $\alpha\geq 1/4$ \cite{AASV21}.
\end{example}

\begin{example}
Suppose that we start with a measure $\mu_0$ on $\binom{[n]}{k}$ that is Strongly Rayleigh \cite[see][for definition]{BBL09}, such as a (symmetric) determinant point process, or the uniform distribution over spanning trees of a graph. Suppose that we partition the ground set into a constant number $c=O(1)$ of parts: $[n]=A_1\cup A_2\cup \cdot A_c$, and fix cardinalities $k_1,\dots, k_c\in \Z_{\geq 0}$, with $k_1+\dots+k_c=k$. Then the partition-constrained version of $\mu_0$ can be defined as
\[ \mu(S)\propto \mu_0(S)\cdot \1\bracks*{\card{S\cap A_i}=k_i\text{ for }i=1,\dots,c}. \]
As long as $c=O(1)$, this distribution $\mu$ will be $\Omega(1)$-fractionally-log-concave \cite{AASV21}. For some discussion of partition-constrained Strongly Rayleigh measures, see \cite{CDKSV16}.
\end{example}

\subsection{Related Work}\label{subsec:relatedwork}

\subsubsection*{Log-Concavity}

Log-concavity has been a well-studied concept in continuous sampling since it captures many common distributions like uniform distributions over convex bodies, and Gaussian distributions. Discrete notions of log-concavity we work with in this paper have been introduced by \cite{Gur09,AOV18,BH19,ALOV19}. The formulation of \cite{ALOV19} is what we refer to in this paper as ``log-concave,'' and is motivated by examples such as the uniform distribution over bases of a matroid and the special subcase of the uniform distribution over spanning trees. 

We have a nearly complete picture for MCMC-based sampling algorithms for homogeneous log-concave distributions. For degree-$k$ distributions, \textcite{ALOV19} analyzed the down-up walk, which occurs between sets of size $k$ and sets of size $(k - 1)$; in the case of matroid bases, this walk is also known as a form of the ``basis exchange walk.'' Furthermore, \textcite{CGM19} proved a Modified Log-Sobolev Inequality (MLSI) for this walk, and \textcite{ALOVV20} further reduced the runtime of sampling by analyzing a warm start to the down-up walk algorithm. Most recently, \textcite{AD20} devised an algorithm for sampling a log-concave distribution when we are given the single-element marginals; we will elaborate upon their contributions more in \cref{subsec:advantage}, where we compare their algorithm to ours. 

\subsubsection*{Intermediate Sampling and Determinantal Point Processes}

A class of domain sparsification algorithms, related to the algorithms we used here, called intermediate sampling was first proposed by \cite{DWH18,Der19} in the context of sampling from Determinantal Point Processes (DPPs, \cite{DM21}), also known as Volume Sampling \cite{DRVW06,DR10,GS12}. DPPs are a family of distributions (a small, but important, subset of distributions with log-concave generating polynomials) which arise for instance when sampling random spanning trees \cite{Gue83}, as well as in randomized linear algebra \cite{DW17,DCMW19}, machine learning \cite{KT11,KT12,DKM20}, optimization \cite{Nik19,DBPM20,MDK20}, and other areas \cite{Mac75,HKPV06,BLMV17}.

The complexity of intermediate sampling for DPPs was further improved by \cite{DCV19,CDV20}, and the approach was extended to DPPs over continuous domains by \cite{DWH19}. Crucially, these algorithms take advantage of the additional structure in DPPs, to enable \emph{distortion-free} intermediate sampling: instead of using a Markov chain, this uses rejection sampling to draw \emph{exactly} from the target distribution. This approach is not possible more generally, since $\mu$ typically does not have a tractable partition function. However, \cite{AD20} showed that the original analysis of distortion-free intermediate sampling can largely be retained for distributions with log-concave generating polynomials, as long as we switch to a Markov chain implementation. 

On the other hand, in this work, we largely abandon the original analysis in favor of a new one which is specific to the Markov chain and requires less precision in marginal estimates. As a result, we show that the preprocessing cost for Markov chain intermediate sampling is substantially smaller than for distortion-free intermediate sampling. This leads to significant improvements in time complexity even for DPPs, e.g., by reducing the preprocessing cost in \cite{DCV19} from $\tilde O(nk^6+k^9)$ to $\tilde O(nk^2+k^3)$, where $\tilde O$ hides polylogarithmic terms.

\subsection{Overview of Techniques}\label{subsec:outline}

Given an entropically independent distribution $\mu: \binom{[n]}{k}\to\R_{\geq 0}$, we first preprocess it using \emph{isotropic transformation}, which is further detailed in \cref{subsec:isotropic}. This converts our distribution $\mu$ into a related distribution $\mu'$ whose single element marginals $\P_{S' \sim \mu'}{i \in S'}$ are approximately uniform. We prove various properties of $\mu'$ in  \cref{prop:near-isotropic}, including the fact that $\mu'$ has ground set size linear in $n$. This preprocessing step may be of independent interest for other discrete sampling problems outside of entropically independent and fractionally log-concave distributions.

After this step, we may assume our distribution $\mu$ has already undergone isotropic transformation. Next, we design a Markov chain $M_\mu^t$ that has $\mu$ as its stationary distribution, where taking a step requires sampling from a sparsified distribution $\nu$. We refer to this algorithm as \emph{Markov Chain Intermediate Sampling}. The benefit of this Markov chain over other natural Markov chains (e.g., down-up random walks \cite[see, e.g.,][]{AASV21}) is that each step requires paying attention only to a subset of elements as opposed to all. We show that $n^{1-\alpha}\cdot \poly(k)$ size is sufficient to ensure that $M_\mu^t$ mixes rapidly. Specifically, in \cref{lem:subsample lower bound}, we prove that a single step of $M_\mu^t$ from any $S \in \supp(\mu)$ satisfies $\norm{P(S, \cdot) - \mu} \leq \frac{1}{4}$. 

The mixing time analysis of $M_\mu^t$ is novel and improves upon the methods used in \cite{AD20}. The improvement is discussed and concretely illustrated with an example distribution in  \cref{subsec:advantage}. The proof of \cref{lem:subsample lower bound} relies heavily on new negative dependence inequalities that are ``average-case'' rather than ``worst-case'', since the worst-case inequalities simply do not hold for fractionally-log-concave distributions. We show that these ``average-case'' inequalities suffice for fast mixing of $M_\mu^t$, thus opening the intermediate sampling framework to wider families of distributions. 

\Tag<noreview>{
\subsection{Acknowledgements}

Nima Anari and Thuy-Duong Vuong are supported by NSF CAREER award CCF-2045354. Elizabeth Yang is supported by the NSF Graduate Research Fellowship under Grant No.\ DGE 1752814.
}
	
	\section{Preliminaries}\label{sec:prelims}

We use $[n]$ to denote the set $\set{1,\dots,n}$. For a set $S$, we use $\binom{S}{k}$ to denote the family of subsets of $S$ of size $k$. For a distribution $\mu$, we use $X\sim \mu$ to denote that $X$ is a random variable distributed according to $\mu$. For a set $U$, we abuse notation and let $X\sim U$ denote $X$ following the \emph{uniform} distribution over $U$.

For a distribution $\mu$ over size $k$ sets and a set $T$ of size potentially larger than $k$, we abuse notation and use $\mu(T)$ to denote:
\[ \mu(T):=\sum_{S\subseteq T}\mu(S). \]

\subsection{Markov Chains and Mixing Time}
\begin{definition}
Let $\mu, \nu$ be two discrete probability distributions over the same event space $\Omega$. The \emph{total variation distance}, or \emph{$\TV$-distance}, between $\mu$ and $\nu$ is given by
$$
\norm{\mu - \nu}_{\TV} = \frac{1}{2} \sum_{\omega \in \Omega} |\mu(\omega) - \nu(\omega)|
$$
\end{definition}

\begin{definition}
Let $P$ be an ergodic Markov chain on a finite state space $\Omega$ and let $\mu$ denote its (unique) stationary distribution. For any probability distribution $\nu$ on $\Omega$ and $\epsilon \in (0,1)$, we define
\[t_\mix(P, \nu, \epsilon) = \min\set{t\geq 0 \mid \norm{\nu P^{t}- \mu}_\TV \leq \epsilon},\]
and
\[t_\mix(P,\epsilon) = \max\set*{t_\mix(P, \1_x, \epsilon) \given x\in \Omega},\]
where $\1_{x}$ is the point mass distribution supported on $x$. 
\end{definition}

We will drop $P$ and $\nu$ if they are clear from context. Moreover, if we do not specify $\epsilon$, then it is set to $1/4$. This is because the growth of $t_{\operatorname{mix}}(P,\epsilon)$ is at most logarithmic in $1/\epsilon$ (cf.~\cite{LP17}). 

The modified log-Sobolev constant of a Markov chain, defined next, provides control on its mixing time. For a detailed coverage see \cite{LP17}.

\begin{definition}
Let $P$ denote the transition matrix of an ergodic, reversible Markov chain on $\Omega$ with stationary distribution $\mu$.

\begin{itemize}
    \item The \emph{Dirichlet form} of $P$ is defined for $f,g \in \Omega \to \mathbb{R}$ by
    \[\mathcal{E}_P(f,g) = \langle f, (I-P)g \rangle_{\mu} = \langle (I-P)f, g \rangle_{\mu}.\]
    \item The \emph{modified log-Sobolev constant} of $P$ is defined to be
\[\rho_0(P) = \inf\left\{\frac{\mathcal{E}_P(f, \log f)}{2\cdot \text{Ent}_{\mu}[f]} : f \colon \Omega \to \mathbb{R}_{\geq 0}, \text{Ent}_{\mu}[f]\neq 0\right\},\]
where 
\[\text{Ent}_{\mu}[f] = \E_{\mu}{f\log f} - \E_{\mu}{f}\log\E_{\mu}{f}.\]
Note that, by rescaling, the infimum may be restricted to functions $f\colon \Omega \to \mathbb{R}_{\geq 0}$ satisfying $\text{Ent}_{\mu}[f]\neq 0$ and $\E_{\mu}{f} = 1$.
\end{itemize}
\end{definition}
The relationship between the modified log-Sobolev constant and mixing times is captured by the following well-known lemma. 

\begin{lemma}[cf.~\cite{bobkov2006modified}]
Let $P$ denote the transition matrix of an ergodic, reversible Markov chain on $\Omega$ with stationary distribution $\mu$ and let $\rho_0(P)$ denote its modified log-Sobolev constant. Then, for any probability distribution $\nu$ on $\Omega$ and for any $\epsilon \in (0,1)$
\[t_{\operatorname{mix}}(P, \nu, \epsilon) \leq \bigg\lceil \rho_0(P)^{-1}\cdot \left(\max_{x\in \Omega}\log\log\left(\frac{\nu(x)}{\mu(x)}\right) + \log\left(\frac{1}{2\epsilon^2}\right) \right) \bigg\rceil.\]
In particular,
\[t_{\operatorname{mix}}(P,\epsilon) \leq  \bigg\lceil \rho_0(P)^{-1}\cdot \left(\log\log\left(\frac{1}{\min_{x\in \Omega}\mu(x)}\right) + \log\left(\frac{1}{2\epsilon^2}\right) \right) \bigg\rceil.\]
\end{lemma}

\begin{theorem}[cf. \cite{LP17}] \label{thm:mixing}
If an irreducible aperiodic Markov chain with stationary distribution $\mu$ and transition matrix $P$ satisfies 
$\norm{P^t (S, \cdot) - \mu}_{TV}  \leq 1/4$ for all $S \in \supp(\mu)$ and some $t \geq 1$, then for any $\epsilon \in (0, 1/4],$
\[t_{\text{mix}}(P, \epsilon) \leq t \log(1/\epsilon).\]
\end{theorem}
\subsection{Fractional Log-Concavity and Entropic Independence}
We recall the notion of fractional log-concavity \cite{AASV21} and entropic independence \cite{AJKPV21}.
\begin{definition}[{\cite{AASV21}}]\label{def:fractional-log-concavity}
	A probability distribution $\mu: \binom{[n]}{k} \to \mathbb{R}_{\geq 0}$ is $\alpha$\emph{-fractionally-log-concave} if $g_\mu(z_1^{\alpha}, \ldots, z_n^{\alpha})$ is log-concave for $z_1, \dots, z_n \in \R_{\geq 0}^n.$
	If $\alpha = 1$, we say $\mu$ is log-concave.
\end{definition}

To define entropic independence we need the definition of the ``down'' operator.
\begin{definition}[Down Operator]
	For $\l\leq k$ define the row-stochastic matrix $D_{k\to \l}\in \R_{\geq 0}^{\binom{[n]}{k}\times \binom{[n]}{\l}}$ by
	\[ D_{k\to \l}(S, T)=\begin{cases}
		0 & \text{if }T\not\subseteq S\\ 
		\frac{1}{\binom{k}{\l}}& \text{otherwise}.
	\end{cases}\]
	Note that for a distribution $\mu$ on size $k$ sets, $\mu D_{k\to \l}$ will be a distribution on size $\l$ sets. In particular, $\mu D_{k\to 1}$ will be the vector of normalized marginals of $\mu$: $(\P{i\in S}/k)_{i\in [n]}$.
\end{definition}

\begin{definition}[{\cite[Definition 2, Theorem 3]{AJKPV21}}]\label{def:ent-ind}
	A probability distribution $\mu: \binom{[n]}{k} \to \mathbb{R}_{\geq 0}$ is $(1/\alpha)$\emph{-entropically-independent} for $\alpha \in (0,1]$, if for all probability distributions $\nu$ on $\binom{[n]}{k}$,
\[ \DKL{\nu D_{k\to 1} \river \mu D_{k\to 1}}\leq \frac{1}{\alpha k}\DKL{\nu \river \mu}.  \]
Or equivalently,
\begin{equation} \label{eq:entropic ind tangent}
    \forall (z_1,\dots, z_n) \in \R^{n}_{\geq 0}: g_{\mu}(z_1^\alpha,\dots, z_n^\alpha)^{1/k\alpha} \leq \sum_{i=1}^{n}p_i z_{i},
\end{equation}
where $p = (p_1,\dots, p_n) := \mu D_{k \to 1}.$ 
\end{definition}
We note that $\alpha$-fractionally log concavity implies $(1/\alpha)$-entropic independence \cite[Theorem~3]{AJKPV21}.

An important part of our sparsification scheme is a process to transform distributions into a ``near-isotropic'' position (defined as having roughly equal marginals) by subdividing the elements of the ground set. More precisely, 
let $\mu$ be a distribution generated by $g_{\mu}(z_1, \dots, z_n),$ then the distribution $\mu'$ obtained by subdividing $z_i$ into $t_i$ copies has generating polynomial
\[g_{\mu'}(z_1^{(1)}, \ldots, z_n^{(t_n)}) = g_\mu \left(\frac{z_1^{(1)} + \ldots + z_1^{(t_1)}}{t_1}, \ldots, \frac{z_n^{(1)} + \ldots + z_n^{(t_n)}}{t_n} \right).\]

Subdivision preserves both entropic independence and fractional log-concavity.

\begin{proposition} \label{prop:entropic-ind-subdivision}
If $\mu$ is $(1/\alpha)$-entropically-independent distribution then $\mu'$ is also $(1/\alpha)$-entropic independence. 
\end{proposition}

\begin{proposition} \label{prop:flc-subdivision}
If $\mu$ is $\alpha$-fractionally-log-concave distribution then $\mu'$ is also $\alpha$-fractionally-log-concave. 
\end{proposition}

We leave the proofs to \cref{sec:proofs}.

	\section{Intermediate Sampling Algorithm}

\subsection{Isotropic Transformation} \label{subsec:isotropic}

 We define, similar to \cite{AD20}, a distribution $\mu$ to be isotropic if for all $i \in [n]$, the marginal probability $\P_{S \sim \mu}{i \in S}$ is $\frac{k}{n}$. We remark that this is only similar in name and spirit, but different in nature, to the analogous notion of isotropy for continuous distributions; the latter is defined based on the covariance matrix of the distribution, while the former is defined based on marginals. In this paper, isotropy captures ``uniformity'' over the elements of $[n]$ in their marginal probabilities. Below we discuss a subdivision process \cite{AD20} that transforms an arbitrary distribution $\mu$ over $\binom{[n]}{k}$ into a distribution $\mu'$ that is nearly-isotropic.
 
 \begin{definition} \label{def:isotropic-transformation}
Let $\mu: \binom{n}{k} \to \mathbb{R}_{\geq 0}$ be an arbitrary probability distribution, and assume that we have estimates $p_1,\dots, p_n$ of the marginals with $p_1+\dots+p_n=k$ and $p_i\geq \Omega(\P_{S\sim \mu}{i\in S})$ for all $i$.
Let $t_i:=\ceil{\frac{n}{k}p_i}$. We will create a new distribution out of $\mu$: For each $i \in [n]$, create $t_i$ copies of the element $i$ and let the collection of all these copies be the new ground set: $U = \bigcup_{i = 1}^{n} \{i^{(1)}, \ldots, i^{(t_i)}\}$. Define the following distribution $\mu': \binom{U}{k} \to \mathbb{R}_{\geq 0}$ from $\mu$:
\[
\mu'\parens*{\set*{i_1^{(j_1)}, \ldots, i_k^{(j_k)}}}:=\frac{\mu(\{i_1, \ldots, i_k\})}{t_1 \cdots t_k}.
\]
We call $\mu'$ the \emph{isotropic transformation} of $\mu$. Another way we can think of $\mu'$ is that to produce a sample from it, we can first generate a sample $\set{i_1, \ldots, i_k}$ from $\mu$, and then choose a copy $i_m^{(j_m)}$ for each element $i_m$ uniformly at random. 
\end{definition}

\begin{remark}
	We note that subdivision or isotropic transformation and external fields behave well together. In particular, a sample from an external field $\lambda$ applied to $\mu'$ can be obtained by first applying an appropriate external field (summing the field values over duplicate elements) to $\mu$ and then replacing each element with a copy of it with probability proportional to $\lambda$. In fact, subdivision is mostly a tool for analysis. In our algorithms, we never have to formally perform subdivision, and we can just sample from distributions defined as $\lambda\star \mu$ for appropriate external fields $\lambda$.
\end{remark}

\begin{remark}
    
To obtain the estimates $\{p_i\}$ for all $i$, we can apply the proof of \cite[Lem.~23]{AD20}, with $\epsilon$ constant, rather than $\epsilon = O(\frac{1}{k})$. This provides a running time reduction for our preprocessing step even in the case of log-concave polynomials.

\end{remark}
There are three desirable properties of $\mu'$ we need to establish for subdivision to be an effective preprocessing step. The first is that subdivision preserves $(1/\alpha)$-entropic independence, which is shown in \cref{prop:entropic-ind-subdivision}. The next is for the marginals $\P_{S \sim \mu'}{i^{(j)} \in S}$ to all be close to $\frac{k}{\card{U}}$ for all $i^{(j)} \in U$; in other words, $\mu'$ is actually close to being isotropic. The last is for $\card{U} \leq O(n)$, so if we ran a sampling algorithm on $\mu'$, the increased size of our ground set does not accidentally inflate our desired asymptotic running times. We remark however, that this last concern can be avoided by simply not running the sampling algorithm on the subdivided distribution, but rather on $\lambda\star \mu$ for an appropriate external field $\lambda$.
\begin{proposition} \label{prop:near-isotropic} 
Let $\mu: \binom{n}{k} \to \mathbb{R}_{\geq 0}$, and let $\mu': \binom{U}{k} \to \R_{\geq 0}$ be the subdivided distribution from \cref{def:isotropic-transformation}. The following hold for $\mu'$:
\begin{enumerate}
\item Near-isotropy: For all $i^{(j)} \in U$, the marginal $\P_{S \sim \mu'}{i^{(j)} \in S} \leq O(k/\card{U})$.
\item Linear ground set size: The number of elements $\card{U} \leq O(n)$. 
\end{enumerate}
\end{proposition}

\begin{proof}
\par First, we verify that $\card{U}$ is at most $O(n)$:
\[
\card{U} = \sum_{i=1}^n t_i\leq \sum_{i=1}^n\parens*{1+\frac{n}{k}p_i}=n+\frac{n}{k}\sum_{i=1}^n p_i=2n.
\]

Next, we check that for any $i^{(j)}$, the marginal probabilities $\P_{S \sim \mu'}{i^{(j)} \in S}$ are at most $O(k/\card{U})$. Here, we interpret the sampling from $\mu'$ as first sampling from $\mu$, and then choosing a copy for each element.
\begin{align*}
\P_{S \sim \mu'}{i^{(j)} \in S} &= \sum_{S \ni i} \P{\text{we chose copy } j \given \text{we sampled } S \text{ from } \mu} \cdot \P{\text{we sampled } S \text{ from } \mu} \\
&= \sum_{S \ni i} \frac{1}{t_i} \cdot \mu(S) = \frac{1}{t_i} \sum_{S \ni i} \mu(S) = \frac{1}{t_i} \cdot \P_{S \sim \mu}{i \in S}.
\end{align*}

Since $t_i\geq \frac{n}{k}p_i\geq \frac{n}{k}\cdot  \Omega(\P_{S\sim \mu}{i\in S})$, we get that
\[ \P_{S\sim \mu'}{i^{(j)}\in S}\leq O\parens*{\frac{\P_{S\sim \mu}{i\in S}}{\frac{n}{k}\cdot \P_{S\sim \mu}{i\in S}}}=O(k/n)\leq O(k/\card{U}). \qedhere \]
\end{proof}

\subsection{Domain Sparsification via Markov Chain Intermediate Sampling}
Here, we first describe, for any general distributions $\mu$, a Markov chain based on generating intermediate samples $T \subseteq [n]$, that mixes to $\mu$. Then, in \cref{lem:subsample lower bound} and \cref{prop:tv-distance-one-step}, we state our main result that for distributions $\mu$ which are $(1/\alpha)$-entropically independent and nearly-isotropic, the size of $T$ only needs to be $n^{1 -\alpha}\cdot \poly(k)$ for the mixing to occur in one step.

Take distribution $\mu:\binom{[n]}{k} \to \R_{\geq 0}$, and consider the following Markov chain $M^t_{\mu}$ defined for any positive integer
$t$, with the state space $\supp(\mu).$ Starting from $S_0 \in \supp(\mu)$, one step of the chain is given by:

\begin{enumerate}
    
    \item Sample $T \sim \binom{[n] \setminus S_0}{t - k}.$ 
    
    \item Downsample $S_1 \sim \mu_{S_0 \cup T}$, where $\mu_{S_0 \cup T}$ is $\mu$ restricted to $S_0 \cup T$, a.k.a.\ $\1_{S_0\cup T}\star \mu$, and update $S_0$ to be $S_1.$
\end{enumerate}
We note that the requirement $S_0 \in \supp(\mu)$ is not strictly necessary for this step to be defined.

\begin{proposition} \label{prop:fast sampling}
For any distribution $\mu: \binom{[n]}{k}\to \R_{\geq 0}$, the chain $M^t_{\mu}$ for $t \geq 2k$ is irreducible, aperiodic and has stationary distribution $\mu.$
\end{proposition}
\begin{proof}
Let $P$ denote the transition probability matrix of $M^t_{\mu}.$
Since $t \geq 2k$, for
any $S, S' \in \supp(\mu)$, there is a positive probability that we sample $T \supseteq S \cup S'$. Thus, we have
$P(S, S') > 0,$ and $P$ is both irreducible and aperiodic.

To derive the stationary distribution, suppose that we perform one step of
the chain starting from $S_0 \sim \mu.$ We first derive the distribution of the intermediate set $R:=S_0 \cup T.$ 

For any $\tilde{R} \in \binom{[n]}{t},$ the probability of sampling $\tilde{R}$ for the intermediate set $R$ is
\[\P{R = \tilde{R}} = \sum_{S_0 \in \binom{\tilde{R}}{k}} \mu(S_0) \cdot \P{T = \tilde{R} \setminus S_0} = \frac{1}{\binom{n-k}{t-k}} \cdot \mu (\tilde{R}) \] 

For any $\tilde{S}_1 \in \supp(\mu)$, the probability of sampling $\tilde{S}_1$ is
\begin{align*}
    \P{S_1 =\tilde{S}_1} &= \sum_{\tilde{R}  \in \binom{[n]}{r}} \P{S_1 =\tilde{S}_1 \given R = \tilde{R} }  \P{R = \tilde{R}}  
    =  \sum_{\tilde{R}  \in \binom{[n]}{t}: \tilde{R} \supseteq \tilde{S}_1}  \frac{\mu(\tilde{S}_1)}{\mu(\tilde{R})} \cdot \frac{1}{\binom{n-k}{t-k}} \cdot \mu (\tilde{R})\\
    &=  \mu (\tilde{S}_1) \sum_{(\tilde{R} \setminus \tilde{S}_1) \in \binom{[n] \setminus  \tilde{S}_1}{t-k} } \frac{1}{\binom{n-k}{t-k}} =\mu(\tilde{S}_1)
\end{align*}
Above, we summed over all $\tilde{R}$ that contain the target set $\tilde{S}_1$.
\end{proof}
The following lemma is the key to analyzing the sampling algorithm, since it quantifies the decrease in $\TV$ distance after running one step of $M^t_\mu$. It will be proven in \cref{subsubsec:subsample lower bound}.
\begin{lemma} \label{lem:subsample lower bound}
Let $\mu: \binom{[n]}{k}$ be a $1/\alpha$-entropically independent distribution. Suppose that for all $i \in [n]$, we have $\P_{S \sim \mu}{i \in S} \leq \frac{Ck}{n}$. Then, for any constant $\epsilon \in (0,\frac{1}{4}],$ and $t = \Omega\parens*{n^{1 -\alpha} (C k^2 \log \frac{1}{1-\epsilon})^{\alpha}}$, the output of a single step of the Markov chain $M^t_{\mu}$ starting from $S_0$ satisfies
\[\forall S \in \binom{[n]}{k}: \P{S_1 = S} \geq \mu(S) (1-\epsilon)\]
\end{lemma}
Recall that if we used the marginal estimates required by \cref{thm:main}, then by applying \cref{prop:near-isotropic} we get an equivalent distribution $\mu'$ on a ground set of size $O(n)$ that satisfies the above assumption of $\P_{S\sim \mu'}{i\in S}\leq Ck/n$ for some $C=O(1)$ (see \cref{lem:subsample lower bound}).

\begin{proposition} \label{prop:tv-distance-one-step} Let $\mu: \binom{[n]}{k}$ be a $1/\alpha$-entropically independent distribution, and let $\epsilon \in (0, \frac{1}{4}]$ be a constant. Suppose $\P_{S \sim \mu}{i \in S} \leq \frac{Ck}{n}$ for all $i$. Choose the intermediate sample size $t$ according to \cref{lem:subsample lower bound}. Then
\[\norm{P(S_0, \cdot) - \mu}_{\TV} \leq \epsilon\]
and the Markov chain $M^t_{\mu}$ mixes to a distribution that has $\TV$ distance $\epsilon' < \epsilon$ from $\mu$ in $O(\log (1/\epsilon'))$ steps.
\end{proposition}
\begin{proof}
The bound on $\TV$ distance follows via 
\[ \|P(S_0, \cdot) - \mu\|_{\TV}  = \sum_{S\in \binom{[n]}{k}:  \P{S_1 = S} < \mu(S) } (\mu(S) -\P{S_1 = S}) \leq \epsilon \sum_{S\in \binom{[n]}{k}:  \P{S_1 = S} < \mu(S) } \mu(S) \leq \epsilon\]
The mixing time bound follows from \cref{thm:mixing}.
\end{proof}

We have shown that $M^t_{\mu}$ is fast mixing (in fact, mixing in one step for appropriately large $t$). Next, we show that for a wide class of distributions, namely, the class of $\alpha$-fractionally-log-concave distributions with $\alpha = \Omega(1)$ \cite[see][for examples]{AASV21}, each step of $M^t_{\mu}$ can be implemented in $\poly(n,k)$ time via a local Markov chain, i.e., the (muti-step) down-up random walk \cite[Def.~1]{AASV21}.
\begin{remark}[Runtime analysis] \label{rem:runtime}
Suppose $\mu$ is $\alpha$-fractionally-log-concave, and we start with $ S_0^{(0)}$ such that $\mu(S_0^{(0)}) \geq 2^{-n^c}$ for some constant $c > 1$ and we run the chain for $\tau$ steps. Then with probability $\geq 1-\tau 2^{-n},$ for all $0\leq i\leq \tau,$ the $i^{th}$-step starting point, denoted by $S_0^{(i)},$ satisfies $\mu(S_0^{(i)}) \geq 2^{-(n^c+ 2n i)}.$ This can be shown via induction on $i.$ Conditioned on $ \mu(S_0^{(i)}) \geq 2^{-n^c - 2ni},$ we have 
\begin{align*}
 \P{\mu(S_0^{(i+1)}) \leq 2^{-(n^c + 2(i+1)n)} } &=\mu(S_0^{(0)} \cup T)^{-1} \sum_{S \subseteq (S_0^{(0)} \cup T):\mu(S) \leq 2^{-(n^c + 2(i+1)n)} } \mu(S) \\
 &\leq_{(1)} \frac{2^{-(n^c + 2(i+1)n)} \cdot 2^n }{\mu(S_0^{(i)})}
 \leq 2^{-n}   
\end{align*}
where in (1) we use the crude bound $\left\lvert \set*{S \subseteq (S_0^{(0)} \cup T):\mu(S) \leq 2^{-(n^c + 2(i+1)n)}} \right\rvert \leq 2^n.$ 

Suppose that this good event happens, i.e.
\[\forall i\in[0,\tau]: \mu(S_0^{(i)}) \geq 2^{-n^c - 2ni}\]
We observe that $\alpha$-fractional-log-concavity is preserved by subdividing \cref{prop:flc-subdivision} and restricting to a subset of the ground set \cite{AASV21}.
In the down-sampling step, we run the (multi-step) down-up walk starting at $S_0^{(i)},$ and use \cite[Thm.~4]{AJKPV21} to bound the runtime. To this end, we need to bound
\begin{align*}
   \mathbb{E}_{T \sim \binom{[n]\setminus S_0^{(i)}}{t-k} }
\left[ \log \parens*{1 + \log \frac{\mu (S_0^{(i)}\cup T ) }{\mu(S_0^{(i)}) } } \right] &\leq_{(1)} \log \parens*{1 + \log \E_{T \sim \binom{[n]\setminus S_0^{(i)}}{t-k} } {\frac{\mu (S_0^{(i)}\cup T ) }{\mu(S_0^{(i)}) }}  }  \\
&= \log \left(1+ \log \frac{1}{\P{S_1 = S_0^{(i)}}} \right)\\
&\leq_{(2)} \log \parens*{1 + \log \frac{1}{\mu(S_1 = S_0^{(i)}) (1-\epsilon)} }\\
&\leq_{(3)} c \log n +\log \tau + \log\log\frac{1} {1-\epsilon}
\end{align*}
where (1) follows from Jensen's inequality for concave function $f(x) = \log (1 + \log(x))$ on $[1, \infty),$ (2) from \cref{lem:subsample lower bound} and (3) from lower bound on $ \mu(S_0^{(i)}).$ The down-sampling then costs \[O\parens*{(t-k)^{\lceil 1/\alpha \rceil} k^{1/\alpha} \left(c \log n + \log \tau + \log\log\frac{1} {1-\epsilon} \right)}\]
and the total runtime is \[O\parens*{\tau (t-k)^{\lceil 1/\alpha \rceil} k^{1/\alpha} \left(c \log n + \log \tau + \log\log\frac{1} {1-\epsilon} \right)}.\]

As a slight optimization, we can replace $(t-k)^{\lceil 1/\alpha\rceil} k^{1/\alpha}$ with $k^{\lceil 1/\alpha\rceil} (t-k)^{1/\alpha}$ when both $\mu$ and its complement $\mu^{\text{comp}}$ are $\alpha$-fractionally-log concave, by down-sampling from  $\mu^{\text{comp}}_{S_0 \cup T}$ then output the complement as $S_1,$ where $\mu^{\text{comp}}: \binom{[n]}{n-k} \to \R_{\geq 0}$ is the complement of $\mu$, defined by $\mu^{\text{comp}}([n] \setminus S) = \mu(S) \forall S \in \binom{[n]}{k}.$ In all important instances of $\alpha$-fractionally-log-concavity, $\frac{1}{\alpha} \in \N$ and this optimization is unnecessary. The bound on total runtime can be simplified into
    $O(n^{1/\alpha-1} \poly(k, \log\frac{1}{\epsilon})).$

\end{remark}

\subsection{Advantage Over Rejection Sampling} \label{subsec:advantage}

\par While we use a similar intermediate sampling framework as \cite{AD20}, our novel analysis of Markov chain intermediate sampling improves the runtime and applies to wider families of distributions. In order to fully understand the advantages realized by our intermediate sampling framework, we first need an overview of a rejection sampling-based implementation of intermediate sampling \cite{Der19}, which inspired the analysis of \cite{AD20}. We then provide an example of $\frac{1}{2}$-log-concave distributions where Markov chain intermediate sampling succeeds using a smaller intermediate sample size than what is required for rejection sampling. 

\par Let $S_0 \in \supp(\mu)$. One step of rejection sampling is given by:
\begin{enumerate}
    \item Sample $T \sim \binom{[n] \setminus S_0}{t - k}$. 
    
    \item Accept the set $S_0 \cup T$ with probability
    $$\frac{\mu(S_0 \cup T)}{\max_{T' \in \binom{[n] \setminus S_0}{t - k}} \mu(S_0 \cup T')}$$
    
    \item Downsample $S_1 \sim \mu_{S_0 \cup T}$. 
\end{enumerate}
The key difference between rejection sampling and our Markov chain intermediate sampling algorithm is the rejection step, which is necessary if we want our chain to mix to the correct stationary distribution $\mu$. In order to have a sufficiently large acceptance probability, and assuming $\mu$ is isotropic, we require that for all $T$,
$$
\mu(T) \leq \left(\frac{t}{n}\right)^k \cdot (1 + \epsilon)^k
$$
Here, $\epsilon$ is a parameter related to the guarantee on $\|P(S_0, \cdot) - \mu\|_{\TV}$. Using this bound, we can ensure that the expected acceptance probability is $1 - O(\epsilon k)$. 

\par This inequality describes a ``worst-case'' condition on $T$. This "worst-case" type analysis originated from earlier works that introduced intermediate sampling for Determinantal Point Processes \cite{Der19}. The proof of our worst-case inequality on $\mu(T)$ relies heavily on the fact that the KL divergence between a log-concave distribution $\mu$ and an arbitrary distribution $\nu$ contracts by a precise amount when applying the down operator $D_{k \to m}$. 
\[
\calD_{\KL}(\nu D_{k \to \ell} \| \mu D_{k \to \ell}) \leq \frac{\ell}{k} \cdot \calD_{\KL}(\nu \| \mu)
\]
This contraction is well-known for log-concave distributions \cite{CGM19}, but does not hold with the factor $\l/k$ for $\alpha$-fractionally-log-concave distributions. On the other hand, the inequality we need to show (from the proof \cref{lem:subsample lower bound}) is ``average-case'' in nature, and when $\mu$ is isotropic, it takes the form:
$$\E_{T \sim \binom{[n]}{t}}{\mu(T)} \leq \left( \frac{t}{n} \right)^k \cdot \frac{1}{1 - \epsilon}$$


To concretely illustrate the advantage of Markov chain intermediate sampling, let us consider an example where the worst-case inequality fails to hold. Suppose that $k=2$, $n$ is even, and $\mu$ samples a set from $\{1,\frac n2+1\}$, $\{2,\frac n2+2\}$,... uniformly at random, so that $\mu(\{i,\frac n2+i\})=\frac 2n$. This distribution is isotropic, $\frac{1}{2}$-sector stable \cite{AASV21}, and $\frac{1}{2}$-fractionally-log-concave, and yet, according to the worst-case analysis, it does not yield enough acceptance probability when $t = o(n)$. For any set $T,$ we have that
\begin{align*}
    \mu(T) \leq \frac{t}{2} \cdot \frac{2}{n} = \frac{t}{n}
\end{align*}
Equality is achieved by selecting a subset $T$ that contains as many  pairs of the form $\{i,\frac n2+i\}$ as possible, i.e., at least $(t-1)/2$. Thus, the worst-case analysis would suggest that no non-trivial intermediate sampling is possible for the distribution $\mu$; this is because $t/n\gg (t/n)^2$ for small values of $t$. 

However, our relaxed average-case analysis captures the fact that realistically, not every element of $T$ will be paired up. In fact, we expect only a constant number of pairs when $t = O(\sqrt{n})$, so for this example, we have:
\[
\E_{T \sim \binom{[n]}{t}}{\mu(T)} \leq C \cdot \frac{2}{n} \leq O\left(\frac{1}{n} \right) = O\left( \frac{t^2}{n^2} \right)
\]

\subsection{Proof of \cref{lem:subsample lower bound}}\label{subsubsec:subsample lower bound}
 
\par In this section, we will prove \cref{lem:subsample lower bound}.
\begin{lemma} \label{lem:subset-probability}
Let $U, V$ be a sets of size $u, v \leq k$ respectively with $ U\cap V = \emptyset.$ We have:
$$
\left(\frac{t - (u + v)}{n - (u + v)} \right)^u \leq \P_{T \in \binom{[n] \setminus V}{t - v}}{U \subseteq T} \leq \left(\frac{t - v}{n - v}\right)^u
$$
\end{lemma}
\begin{proof}
Since we are sampling the elements of $T$ uniformly at random from $[n]$,
\begin{align*}
\P_{T \in \binom{[n] \setminus V}{t - v}}{U \subseteq T} = \frac{\binom{n - v - u}{t - v - u}}{\binom{n - v}{t - v}} = \frac{(t - v)(t - v - 1) \cdots (t - v - u + 1)}{(n - v)(n - v - 1) \cdots (n - v - u + 1)} \leq \left(\frac{t - v}{n - v} \right)^u 
\end{align*}
Similarly, we also have:
$$\P_{T \in \binom{[n] \setminus V}{t - v}}{U \subseteq T} \geq \left(\frac{t - v - u + 1}{n - v - u + 1} \right)^u \geq \left(\frac{t - (u + v)}{n - (u + v)} \right)^u
$$
\end{proof}

\begin{proof}[Proof of \cref{lem:subsample lower bound}]
Let $R = S_0 \cup S$, and let $r = |S_0 \cup S|$. Note that $\abs{S \setminus S_0} = r-k$ and
\begin{align*}
    \P{S_1 = S} &= \E_{T \sim \binom{[n] \setminus S_0 }{t - k}} { \frac{\mu(S)}{\sum_{S' \subseteq (S_0 \cup T)} \mu(S')} \big| S \subseteq T } \cdot \P_{T \sim \binom{[n] \setminus S_0}{t - k}}{(S \setminus S_0) \subseteq T} \\
    &= \E_{T' \sim \binom{[n] \setminus R}{t -r}} { \frac{\mu(S)}{\sum_{S' \subseteq (R \cup T')} \mu(S')} } \cdot \P_{T \sim \binom{[n] \setminus S_0}{t - k}}{(S\setminus S_0) \subseteq T} \\
    &\geq_{(1)} \frac{\mu(S)}{\E_{T' \sim \binom{[n] \setminus R}{t - r}}{\sum_{S' \subseteq (R \cup T')} \mu(S')}} \cdot \P_{T \sim \binom{[n]\setminus S_0}{t - k}}{(S\setminus S_0) \subseteq T}\\
    &\geq_{(2)}   \frac{\mu(S)}{\E_{T' \sim \binom{[n] \setminus R}{t - r}}{\sum_{S' \subseteq (R \cup T')} \mu(S')}} \cdot  (\frac{t-r}{n-r})^{r-k}
\end{align*}
Inequality (1) is an application of Jensen's inequality to the function $f(x) = \frac{c}{x}$, which is convex when $x > 0$. Inequality (2) use \cref{lem:subset-probability} with $U = S\setminus S_0$ and $V = S_0.$

Now if we bound $\E_{T' \sim \binom{[n] \setminus R}{t - r}}{\sum_{S' \subseteq (R \cup T')} \mu(S')} $ by $  (\frac{t-r}{n-r})^{r-k} \frac{1}{1-\epsilon}$ then we are done.
\begin{align*}
\E_{T' \sim \binom{[n] \setminus R}{t - r}} {\sum_{S' \subset (T' \cup R)} \mu(S') } = \sum_{S' \in \binom{[n]}{k}} \mu(S') \cdot \P_{T' \sim \binom{[n] \setminus R}{t - r}}{(S' \setminus R) \subseteq T']} 
\leq \sum_{S' \in \binom{[n]}{k}} \mu(S') \cdot \left( \frac{t - r}{n - r} \right)^{|S' \setminus R|} 
\end{align*}
In the very last line, we applied \cref{lem:subset-probability} with $U = |S' \setminus R|$ and $V = R.$

If we set $z_i =\begin{cases} (\frac{n - r}{t - r})^{1/\alpha} & \text{ if } i \in (S_0 \cup S) \\ 1 &\text{ else}\end{cases}$, then we can rewrite
\begin{align*}
\sum_{S' \in \binom{[n]}{k}} \mu(S') \cdot \left( \frac{t - r}{n - r} \right)^{|S' \setminus (S_0 \cup S)|} &= \left( \frac{t - r}{n - r} \right)^k \sum_{S' \in \binom{[n]}{k}} \mu(S') \cdot \left( \frac{n - r}{t - r} \right)^{|S' \cap (S_0 \cup S)|} \\
&= \left( \frac{t - r}{n - r} \right)^k \cdot g_\mu\parens*{z_1^{\alpha}, \ldots, z_n^{\alpha} }
\end{align*}

Applying \cref{eq:entropic ind tangent} to $g_\mu(z_1^{\alpha}, \ldots, z_n^{\alpha})$ and noting that $p_i = \frac{\P_{S \sim \mu}{i \in S}}{k} $, we obtain
\begin{align*}
     g_\mu\parens*{z_1^{\alpha}, \ldots, z_n^{\alpha} } &\leq \left( \sum_{i = 1}^n \frac{\P_{S \sim \mu}{i \in S}}{k} \cdot z_i \right)^{k \alpha} \\
      \log g_\mu\parens*{z_1^{\alpha}, \ldots, z_n^{\alpha} } &\leq k \alpha \log \left( \sum_{i = 1}^n  \frac{\P_{S \sim \mu}{i \in S}}{k}  \cdot z_i \right) \\
      &\leq_{(1)}  k \alpha \parens*{-1 + \sum_{i = 1}^n  \frac{\P_{S \sim \mu}{i \in S}}{k}  \cdot z_i  }\\
      &=_{(2)} k \alpha \parens*{\sum_{i = 1}^n  \frac{\P_{S \sim \mu}{i \in S}}{k}  \cdot (z_i - 1)}\\
      &= \alpha \sum_{i = 1}^n  \P_{S \sim \mu}{i \in S}  \cdot (z_i - 1)
\end{align*}
where in (1) we use $\log x \leq x -1$ for $x\in (0,\infty)$ and in (2) we use $\sum_{i = 1}^n  \frac{\P_{S \sim \mu}{i \in S}}{k} = 1 .$

Substitute $z_i$ as specified above into the final inequality, we get
\begin{align*}
    \log g_\mu\parens*{z_1^{\alpha}, \ldots, z_n^{\alpha} }     &\leq \alpha \sum_{i \in (S_0 \cup S)} \frac{Ck}{n} \left(\frac{n - r}{t - r} \right)^{1/\alpha}\\
    &= \frac{C\alpha k r }{n}\left(\frac{n - r}{t - r}\right)^{1/\alpha}\\
    &\leq \frac{2C k^2  }{n} \left(\frac{n - r}{t - r}\right)^{1/\alpha}.
\end{align*}
\end{proof}

	\section{Lower Bound on Intermediate Sampling}\label{sec:lowerbound}

We first show that the dependence of our sparsification analysis on $n$ is optimal. Consider the uniform distribution $\mu_0$ over singletons of a ground set of $n/k$ elements. Any distribution on singletons is log-concave as the generating polynomial is linear and thus log-concave. Now apply the construction of \cref{ex:blowup} with $m=k$ to $\mu_0$ in order to obtain a new distribution $\mu$ on $\binom{[n]}{k}$. This distribution is uniform over parts of a particular partition of the ground set $[n]$ into $n/k$ sets $S_1,\dots,S_{n/k}$. As such, the distribution is also isotropic. Note that this distribution is also $1/k$-entropically independent.

If we sample a uniformly random set $T$ of size $t$, then the chance that $S_i$ is contained in $T$ can be upperbounded by
\[ \binom{n-k}{t-k}/\binom{n}{t}\simeq (t/n)^k. \]
Thus the chance that any of the $S_i$ are contained in $T$ can be upperbounded (via a union bound) by roughly
\[ n\cdot (t/n)^k. \]
Thus, as long as $t\ll n^{1-1/k}$, the above is negligible. Without having any $S_i$ in the support with high probability, we obviously cannot faithfully produce a sample of $\mu$ from subsets of $T$.

Next we construct an example showing that even higher-order marginals cannot remove the dependence on $n$ for entropically independent distributions (in sharp contrast with \cref{conj:high-order}). Our construction is based on Reed-Solomon codes.

\begin{lemma}
Let $q$ be a prime number and $\F_q$ the finite field of size $q$. Fix $k$ points $\set{x_1,\dots,x_k}\subseteq \F_q$ where $k$ is a constant and choose a set of $k$ random permutations from $\mathbb{F}_q$ to $\mathbb{F}_q$ and call them $\pi_1, \ldots, \pi_k$. Let $\mu:\binom{\Omega}{k} \to \R_{\geq 0}$ be the uniform distribution over sets $\set*{(x_i, y_i) \given i \in [k]}$ s.t.\ $p(x_i) = \pi_i(y_i)$ for some polynomial $g$ of $\deg(p) \leq d$. Note that the ground set $\Omega$ is $\set{x_1,\dots,x_k}\times \F_q$. Then
\begin{enumerate}
    \item $\mu$ satisfies $(1/\alpha) $-entropic independence with $\alpha = \frac{d+1}{k}.$
    \item Any domain sparsification scheme to sample from $\mu$ requires $t = \tilde{\Omega}(n^{1-\alpha}),$ even when we are allowed to sample higher order marginals.
\end{enumerate}
\end{lemma}
\begin{proof}
The distribution $\mu D_{k \to (d+1)}$ is uniform over $\{(x_j, y_j): j \in J \subseteq [k], |J| = (d + 1)\}$, because for any such set, there exists a unique polynomial $p$ of degree at most $d$ such that $p(x_j) = \pi_j(y_j)$ for all $j\in J$. Thus, high-order marginals, up to order $d+1$, are independent of the choice of permutations $\pi_1,\dots,\pi_k$.


The support of $\mu D_{k\to (d+1)}$ forms the basis of a partition matroid: for each $x=x_i$, we have a block consisting of all points $\{(x, y) : y \in \mathbb{F}_q\}$, and for each set in the support of $\mu D_{k \to (d+1)}$, we have at most one element per block. Since we have a uniform distribution over matroid bases, $\mu D_{k \to (d+1)}$ is log-concave, and thus it satisfies $1$-entropic independence. We use this to upper bound $\DKL{\nu D_{k \to 1} \river \mu D_{k \to 1}}$, and from here, conclude $\frac{d + 1}{k}$-entropic independence of $\mu$:
\begin{align*}
    \DKL{\nu D_{k \to 1} \river \mu D_{k \to 1}} &= \DKL{(\nu D_{k \to (d + 1)}) D_{(d + 1) \to 1} \river (\mu D_{k \to (d + 1)}) D_{(d + 1) \to 1}} \\
    &\leq \frac{1}{d + 1} \cdot \DKL{\nu D_{k \to (d + 1)} \river \mu D_{k \to (d + 1)}} \\
    &\leq \frac{1}{d + 1} \cdot \DKL{\nu \river \mu} = \frac{1}{\frac{d + 1}{k} \cdot k} \cdot \DKL{\nu \river \mu}
\end{align*}
The second line follows from $\mu D_{k \to (d+1)}$ satisfying $1$-entropic independence, and the third line comes from the data-processing inequality. 

\par We now prove that for all $t \leq o\left(n^{1 - \alpha} \right)$, no domain sparsification scheme exists, even with access to higher order marginals. For $d' \leq d + 1$, the distribution $\mu D_{k \to {d'}}$ is uniform over the size $d'$ independent sets of the partition matroid defined above. One consequence of the independence of high-order marginals from the choice of permutations $\pi_1,\dots,\pi_k$ is that the higher order marginals do not provide any information about the identity of the distribution $\mu$. 

Suppose that we choose our sample in domain sparsification from a distribution whose ground set is the sparse subset $T$. We want an upper bound on the probability (over the choice of permutations) that $T$ contains some $S \in \supp(\mu)$. 

In order for $T$ to contain a valid $S$, there must be some subset in $S \in \binom{T}{k}$ associated to a degree $\leq d$ polynomial $p$ satisfying $p(x_i) = \pi_i(y_i)$. However it is easy to see that for any particular set $S$ the probability (over the choice of permutations) that $S$ is in the support of $\mu$ is $\simeq 1/q^{k-d-1}$.

We can upper bound $\P{S \subseteq T \text{ for some }S\in \supp(\mu)}$ by a union bound as follows:
\begin{align*}
    \binom{t}{k} \cdot \frac{1}{q^{k-d-1}} \leq \frac{t^k}{q^{k-d-1}}
\end{align*}
This implies that for any $t \leq o(q^{(k-d-1)/k}) \leq o(n^{(k-d-1)/k})$, the probability of containing a set in the support is negligible. Note that we have $\alpha = \frac{d+1}{k}$, so $1 - \alpha = \frac{k-d-1}{k}$, which completes the lower bound.
\end{proof}
	
	\section{Missing Proofs}\label{sec:proofs}

In this section we prove \cref{prop:entropic-ind-subdivision,prop:flc-subdivision}. To complete the proofs, we need to understand a few results concerning the correlation matrix of a distribution $\mu$ to get an alternate characterization of $\alpha$-fractional-log-concavity (\cref{thm:correlation matrix vs FLC}).

\begin{definition}
For distribution $\mu: \binom{[n]}{k} \to \R_{\geq 0},$ let the \textit{correlation matrix} $\corMat_{\mu} \in \R^{n\times n}$  be defined by
\[\corMat_{\mu} (i,j) = \begin{cases} 1 -\P{i} &\text{ if } j=i \\ \P{j \given i} - \P{j} &\text{ else}\end{cases} \]
where $\P{j\given i} = \P_{T \sim \mu}{ j \in T \given i \in T}, \P{j} = \P_{T \sim \mu}{j \in T}$.
\end{definition}

\par For a distribution $\mu$ on $\binom{[n]}{k}$ and $\lambda = (\lambda_1,\dots, \lambda_n) \in \R^{n}_{>0}$, the \emph{$\lambda$-external field} applied to $\mu$ is also a distribution on $\binom{[n]}{k}$, denoted by $\lambda \ast \mu$, given by
\[\mathbb{P}_{\lambda \ast \mu}[S] \propto \mu(S)\cdot \prod_{i \in S}\lambda_i.\]
Note that for any $(z_1,\dots, z_n) \in \R^{n}_{\geq 0}$, 
\[g_{\lambda \ast \mu}(z_1,\dots, z_n) \propto g_{\mu}(\lambda_1 z_1,\dots, \lambda_n z_n).\]

\begin{theorem} [{\cite[Lemma~67, Remark~68]{AASV21}}] \label{thm:correlation matrix vs FLC}
A polynomial $g$ is $\alpha$-fractionally-log-concave if and only if the largest eigenvalue of $\lambda * \mu$'s correlation matrix satisfies $ \lambda_{\max}(\corMat_{\lambda \ast \mu} ) \leq \frac{1}{\alpha}$ for all external fields $\lambda \in \R_{\geq 0}^n$.
\end{theorem}

Now that we understand an alternate characterization of $\alpha$-fractional-log-concavity, we can use it to complete the proof of \cref{prop:flc-subdivision}.

\begin{proof}
	By induction, we only need to prove the statement when one element is subdivided. So we need to show that if $g:=g_\mu$ is fractionally log-concave, then
\[h(z_1^{(1)}, \ldots , z_1^{(k_1)}, z_2, \dots, z_n): =g(z_1^{(1)} + \ldots + z_1^{(k_1)}, z_2, \dots, z_n) \text{ is $\alpha$-fractionally log-concave} \]

By \cref{thm:correlation matrix vs FLC}, this is equivalent to \[\lambda_{\max}(\corMat_{\lambda' \ast \mu'} ) \leq \frac{1}{\alpha} \; \; \forall \vec{\lambda} = (\lambda_1^{(1)}, \dots, \lambda_1^{(k_1)}, \lambda_2, \dots, \lambda_n)\in \R_{\geq 0}^{n+k_1-1},\] where $\mu'$ is the distribution generated by $h.$ Without loss of generality, we assume $\sum_{j=1}^{k_1} \lambda_1^{(j)} = 1.$

\par Let $\Psi$ and $\Psi'$ be the correlation matrix of $\lambda \ast \mu$ and $\lambda' \ast \mu',$ where $\lambda = (1, \lambda_2, \dots, \lambda_2)\in \R_{\geq 0}^n.$ We want to relate the eigenvalues of $\Psi'$ to $\Psi,$ thereby showing that $\Psi'\leq \frac{1}{\alpha}.$ 

Note that
\begin{align*}
    &\P_{\lambda' \ast \mu'}{1^{(i_1)}} = \lambda_1^{(i_1)}\P_{\mu}{1}    \text{ and } \P_{\lambda' \ast \mu'}{1^{(i_1)} \given 1^{(i_2)}} = 0 \\
    &\forall j\neq 1: \P_{\lambda' \ast \mu'}{1^{(i_1)} \given j} = \lambda_1^{(i_1)}\P_{\lambda \ast \mu} { 1 \given j} \text{ and }
    \P_{\lambda' \ast \mu'}{j } = \P_{\lambda \ast \mu}{j} \text{ and } \P_{\lambda' \ast \mu'}{j \given 1^{(i_1)}} = \P_{\lambda\ast \mu}{j \given 1} 
\end{align*}

Let $\vec{v}^1, \dots, \vec{v}^n$ be the orthogonal basis  $\langle \cdot, \diag(\P_{\lambda \ast
\mu}{i^{(j_i)} })$ of eigenvectors of
$\Psi$ with corresponding eigenvalues $\rho_1 \geq \rho_2 \geq \dots \geq \rho_n.$ 
It is easy to check that for each $\vec{v}^i,$ the vector \[\vec{w}^i : =(\underbrace{v^i_1, \dots, v^i_1}_{k_1 \text{ times}}, v^i_2, \dots, v^i_n  \}\]
is an eigenvector  of $\Psi'$ associated with eigenvalue $\rho_i.$  In addition, $\Psi'$ also has eigenvectors $\vec{w}^{n+i}:= (u^i_{1}, \dots, u^i_{k_1},\underbrace{0, \dots, 0}_{n-1 \text{ times}} ) $ associated with eigenvalue $ 1,$ where $\set*{\vec{u}^i}_{i=1}^{k_1-1}$ forms an orthogonal basis wrt $\langle \cdot, \diag(\P_{\lambda \ast
\mu}{i})$ of the vector space $\set*{\vec{u} \in \R^{k_1} \given\vec{u} \perp (\lambda_1^{(i)})_{i=1}^{k_1}}.$  Observe that $\set*{\vec{w}^i}_{i=1}^{n+k_1-1}$ forms an orthogonal basis wrt $\langle \cdot, \diag(\P_{\lambda' \ast
\mu'}{i^{(j_i)} })$ of eigenvectors of $\Psi',$ thus the spectrum of $\Psi'$ is
\[\set*{\rho_1, \dots, \rho_n} \cup \set*{1^{(k_1-1)}}\]
Thus, $\lambda_{\max}(\Psi') \leq \max(\rho_1, 1) \leq\frac{1}{\alpha}$,
 by \cref{thm:correlation matrix vs FLC} applied to $\mu$.
\end{proof}

Next we prove \cref{prop:entropic-ind-subdivision}.

\begin{proof}
	We have a characterization of entropic independence as
	\[ g_\mu(z_1^\alpha,\dots,z_n^\alpha)^{1/k\alpha}\leq \sum_{i\in S} p_i z_i, \]
	where $p_i=P_{S\sim \mu}{i\in S}/k$ are the normalized marginals of $\mu$. Subdivision replaces the generating polynomial by a new polynomial, where $z_i$ is replaced by an average of $t_i$ new variables. It is easy to see that the normalized marginal of each of the new variable is simply $p_i/t_i$. Thus $z_i$ in the r.h.s.\ of the above expression also gets replaced by the same average of the $t_i$ new variables. So the proof is just a matter of plugging in $(z_i^{(1)}+\dots+z_i^{t_i})/t_i$ for $z_i$ in the above ineuality.
\end{proof}

	\section{Conclusion}

A natural direction to extend our research is to better understand the trade-offs between the run-time of approximate sampling algorithms and how much information we are given in the form of marginals for $\alpha$-fractionally-log-concave distributions. Recall that $\alpha$-fractional-log-concavity implies $1/\alpha$-entropic independence, but still captures many interesting distributions (see \cref{subsec:applications}), so while we exhibit a lower bound against domain sparsification for $1/\alpha$-entropically independent distributions, even with access to higher-order marginals, the stronger assumption of fractional log-concavity may avoid this barrier.

For a $k$-uniform distribution, the entire distribution itself is indeed captured by the order-$k$ marginals. So if one is to believe that $\alpha$-fractionally log-concave distributions behave like distributions over $1/\alpha$-sized sets, does that mean order-$1/\alpha$ marginals are sufficient to get rid of all dependence on $n$ in domain sparsification? For $\alpha=1$, the answer is affirmative as was shown by \textcite{AD20}. A fascinating open question is whether we can extend this to fractionally log-concave polynomials as suggested by \cref{conj:high-order}.
	
	\PrintBibliography
\end{document}